\documentclass[pra,aps,onecolumn,floatfix,superscriptaddress,notitlepage]{revtex4-1}
\usepackage{amsthm}
\usepackage{amsmath,bm}
\usepackage{amssymb}
\usepackage{braket}
\usepackage{bbold}
\usepackage{graphicx}
\usepackage{subfigure}

\usepackage{appendix}
\usepackage[T1]{fontenc}
\usepackage{color}
\usepackage[utf8]{inputenc}
\usepackage[linesnumbered,ruled,vlined]{algorithm2e}
\usepackage{multirow}
\usepackage{hyperref}
\usepackage{pdfsync}
\newcommand{\nc}{\newcommand}
\nc{\tr}{{\operatorname{Tr}}}
\newtheorem{theorem}{Theorem}

\newtheorem{lemma}[theorem]{Lemma}
\newtheorem{proposition}[theorem]{Proposition}
\newcommand{\RomanNum}[1]{\textsc{\romannumeral #1}}

\usepackage{natbib}

\begin{document}
\title{Classifying Informative Subsets in Qudit Encrypted Cloning via Congruence Systems}
\author{Chen-Ming Bai}
\email{baichm@stdu.edu.cn}
\author{Xin-Liang Zhou}
\affiliation{Department of Mathematics and Physics, Shijiazhuang Tiedao University,
Shijiazhuang, 050043, China}
\author{Yu Luo}
\email{penroseluoyu@gmail.com}
\affiliation{School of Artificial Intelligence and Computer Science, Shaanxi Normal University, Xi'an, 710062, China}

\begin{abstract}
Encrypted cloning offers a means of introducing redundancy into quantum storage while respecting the no-cloning theorem: an unknown state is encoded into multiple signal-noise pairs, and only authorized subsets can recover the original information. However, the leakage properties of unauthorized subsets, particularly for higher-dimensional systems (qudits), have remained unexplored. In this work, we systematically classify the informative subsets of the storage register in the qudit encrypted-cloning protocol. We focus on unauthorized subsets of size $n$ that contain exactly one qudit from each signal-noise pair. We show that the presence or absence of information leakage is determined by the solution set of a system of congruences whose coefficients depend on the dimension $d$ and on the numbers of signal and noise qudits in the subset. The reduced state is completely uninformative if and only if the congruence system admits only the trivial solution; otherwise, it retains a residual dependence on the input state through specific generalized Pauli operators. Low-dimensional examples ($n=1,2,3$) are worked out explicitly, and the complete classification is expressed in terms of a greatest-common-divisor condition. Our results extend the parity-based classification known for qubits ($d=2$) to arbitrary finite dimensions, revealing a dimension-dependent boundary of confidentiality in encrypted cloning. Finally, the modified protocol makes the joint quantum--classical state of every unauthorized subset independent of the input. For the aligned unauthorized subsets containing exactly one qudit from each signal--noise pair, the quantum marginal is, in addition, maximally mixed.
\end{abstract}

\maketitle
\section{Introduction}
\label{sec:intro}
In classical information processing, copying bits is straightforward. In quantum mechanics, however, the no-cloning theorem~\cite{wootters1982single,dieks1982communication}--a direct consequence of linearity and unitarity--forbids the perfect duplication of an unknown state. Alongside the uncertainty principle and the no-discrimination theorem, no-cloning provides the foundation for the unconditional security of quantum cryptography~\cite{bennett1992quantum, gisin2002quantum, pirandola2020advances}, quantum key distribution~\cite{renner2008security, li2025microsatellite,Wiesemann2026consolidated} and quantum error correction~\cite{wang2026demonstration,aydin2026quantum}.
Following this theorem, relaxations have led to a theory of quantum cloning, encompassing approximate cloning~\cite{buvzek1996quantum, hillery1997quantum},  probabilistic cloning~\cite{duan1998probabilistic} and virtual cloning~\cite{bi2025virtual}.

In 2026, Yamaguchi and Kempf proposed a protocol called encrypted cloning~\cite{yamaguchi2026encrypted}, which successfully enables redundant cloning of an unknown qubit while remaining fully compatible with the no-cloning theorem. In this protocol, the original qubit interacts unitarily with several pre-prepared Bell pairs, generating multiple encrypted clones (signal qubits) together with corresponding noise qubits. Individually, each signal qubit carries no information about the original state; however, by using all noise qubits as a one-time key, the original state can be fully recovered. This protocol provides, for the first time, a way to achieve "clonable but decryptable-only-once" quantum redundancy, opening up promising applications in quantum storage~\cite{huang2024performing}, quantum secret sharing~\cite{hillery1999quantum, zhang2025device, di2025security}, and quantum communication~\cite{usenko2026continuous}.

Building on this work, Cear{\v{a}} extended the encrypted cloning protocol to quantum systems of arbitrary finite dimension (qudits)~\cite{cear2026clon}. The author addressed the non-unitarity challenge that emerges in higher dimensions--where generalized Pauli operators are no longer Hermitian--by constructing encryption operators based on a constant-amplitude zero-autocorrelation (CAZAC) sequence. The corresponding decryption operators are also provided, and their circuit implementation complexity is analyzed. This generalization demonstrates that the encrypted cloning mechanism is intrinsically dimension-independent and can be realized in qudit systems as well.

However, the confidentiality guarantee of the encrypted cloning protocol does not imply that every unauthorized subset is completely non-informative. Recently, Gianini et al.~\cite{gianini2026encrypted} showed that, for qubits, certain unauthorized subsets of the storage register (i.e., those failing the full recovery condition) can still retain partial information about the original state. By systematically analyzing the subset structure of the storage register, they found that information leakage depends solely on the parity of the number of signal qubits in the subset, and that only the
$y$-component of the Bloch vector is leaked. This finding reveals a structural limitation of confidentiality in encrypted cloning and demonstrates that "authorization" and "information leakage" are not strict opposites.

Several recent results further develop the operational and structural theory of encrypted cloning. An implementation on IBM Heron-R2 processors, involving up to 154 qubits, found that encrypted cloning remains usable under hardware noise and can be composed in parallel, in series, and with intervening operations while preserving pre-existing entanglement~\cite{yamaguchi2026experimental}. Lim and Lo related qudit encrypted cloning to absolutely maximally entangled states and placed it within the framework of quantum secret sharing~\cite{lim2026ame}. For the qubit protocol, Gianini et al. completed the leakage classification for subsets that also contain the transformed source qubit~\cite{gianini2026full}, and subsequently showed how secret-sharing access structures generate encrypted-cloning schemes with nonthreshold, asymmetric, and ramp authorization patterns~\cite{gianini2026access}. The same redundancy has also been used for state-blind relational fault diagnosis~\cite{gianini2026diagnostic}. Most recently, Roy and Gupta characterized the multipartite resource states that permit exact encrypted recovery, deriving entanglement lower bounds and exact graph-state criteria~\cite{roy2026resources}. These developments motivate a precise treatment of unauthorized marginals: recovery structure, resource entanglement, and implementation robustness do not by themselves determine whether a nonrecovering subset is input-independent.

Motivated by these results, in this paper, we systematically investigate information leakage in qudit-based encrypted cloning. Extending the classification framework of Gianini et al. to arbitrary dimensions, we focus on unauthorized subsets that contain exactly one qudit from each signal-noise pair, i.e., aligned subsets of size
$n$. Through analytical computation of the reduced density matrix and exploitation of the algebraic structure of generalized Pauli operators, we show that the occurrence of leakage no longer depends on simple parity conditions. Instead, it is governed by the solution set of a system of congruences, which ultimately reduces to a greatest-common-divisor condition involving the dimension $d$, the number $p$ of signal qudits, and the number
$(n-p)$ of noise qudits.

Specifically, for an aligned subset consisting of
$p$ signal qudits and $n-p$ noise qudits, we prove that the reduced state is completely uninformative (maximally mixed) if and only if
$\gcd\!\left(d,p(n-p+1)-1\right)=1$. Otherwise, the subset retains a residual dependence on the input state, which manifests through specific generalized Pauli expectations; such subsets are therefore partially informative. For
 $d=2$, our classification reduces exactly to the parity-based criterion of Gianini et al., thereby providing a rigorous generalization to higher dimensions and revealing a dimension-dependent confidentiality boundary in encrypted cloning. Finally, by leveraging group-average orthogonality and classical secret sharing, we eliminate all input dependence for every unauthorized subset; maximal mixing is obtained for the aligned unauthorized subsets of size $n$.

The remainder of the paper is structured as follows. In Section \ref{sec:Preliminary}, we review the qudit encrypted cloning protocol and introduce the necessary notation for generalized Pauli operators. In Section \ref{sec:Classification}, we demonstrate the leakage mechanism via low-dimensional examples ($n=1,2,3$) and give the general classification of subsets of the encrypted-clone storage register. In Section \ref{sec:Improvement}, we use group-average orthogonality to remove all input dependence from every unauthorized subset and prove maximal mixing for the aligned subsets of size $n$. Section \ref{sec:conclusion} concludes the paper with a discussion of the security implications of our results.
\section{Preliminaries}
\label{sec:Preliminary}
\subsection{Generalized Pauli operators}
For any natural number \(d\), let \(\mathbb{Z}_d= \{0, 1, \cdots, d-1\}\) be the set of congruence classes modulo \(d\). Then we consider a \(d\)-dimensional Hilbert space $\mathcal{H}$ with a computational basis \(\{ |i\rangle \mid i \in \mathbb{Z}_d \}.\)
Furthermore, we define the shift operator \(X_d\) and the phase operator \(Z_d\) on $\mathcal{H}$ as follows \cite{schwinger1960unitary}:
\begin{equation}
\label{eq:xzoperator1}
   X_d=\sum_{k=0}^{d-1}\ket{k+1}\bra{k},\quad Z_d=\sum_{k=0}^{d-1}\omega^k\ket{k}\bra{k}.
\end{equation}
where "+"  is expressed as modulo $d$. These operators act on the computational basis as
\begin{equation}
\label{eq:xzoperator2}
   X_d|k\rangle = |k+ 1\rangle,\quad Z_d|k\rangle =\omega^k|k\rangle,
\end{equation}
for \(k \in \mathbb{Z}_d\), with \(\omega = e^{2\pi {\rm i}/d}\). Therefore, $X_d^\dagger=X_d^{\mathsf T}= X_d^{-1}$, $Z_d^{\mathsf T} = Z_d$, and $Z_d^\dagger = Z_d^{-1}$.

Moreover, it can be easily checked that \(X_d^d = Z_d^d = I_d\), where
$d$ is the dimension of the Hilbert space.
And these operators satisfy the commutation relation $X_dZ_d = \omega^{-1}Z_dX_d$,
which implies
\begin{equation}
    \label{eq:xzoperator3}
    X_d^k Z_d^l = \omega^{-kl} Z_d^l X_d^k,
\end{equation}
for all \(k,l \in \mathbb{Z}_d\).
\subsection{The generalized Pauli operators-based encoding}
In this section, we review the encrypted-cloning protocol for qudits introduced in \cite{cear2026clon}, focusing on the single-input qudit case and the notation required for leakage analysis. Let $A$
 denote the input qudit, which is prepared in an unknown pure state $\ket{\psi}_A$.
To generate \(n>1\) encrypted clones, the protocol introduces \(n\) $d$-dimensional maximal entangled state:
\begin{equation}
    \label{eq:messtate1}
  \ket{\Phi_d}_{S_iN_i}=\frac{1}{\sqrt{d}}\sum_{k=0}^{d-1}\ket{k}_{S_i}\ket{k}_{N_i},
\end{equation}
where $i=1,2,\cdots,n$, \(S_i\) denotes the \(i\)-th signal qudit and \(N_i\) represents the corresponding noise qudit.

The encoding for encrypted cloning is implemented by a unitary transformation that jointly acts on the input qudit \(A\)
 and the signal qudits \(S_1,\cdots,S_n\), while leaving the noise qudits \(N_1,\cdots,N_n\) unaffected. To encode a pure state $\ket{\psi}$ into $n$ encrypted clones, Cear{\v{a}} introduces the following unitary:
\begin{equation}
    \label{definition_enc_mine_final}
    U_{{\rm enc}_{d}}^{(n)} = \frac{1}{d}\sum_{k,l=0}^{d-1}c_{kl} ({X_{ d}^{(A)}})^k({Z_{ d}^{(A)}})^{l} \otimes\left(\bigotimes_{i=1}^{n}({X_d^{(S_i)}})^{k}({Z_d^{(S_i)}})^{l}\right),
\end{equation}
where \(c_{kl}= \exp(\frac{-{\rm i}\pi k(k+\epsilon_d)}{d})\cdot \exp(\frac{-{\rm i}\pi l(l+\epsilon_d)}{d})\), with the parity parameter $\epsilon_d:=d\bmod 2\in\{0,1\}$, and \(|c_{kl}|=1\) for every $k,l\in\mathbb Z_d$.

As given in Ref.~\cite{cear2026clon}, the initial state for encrypted-cloning encoding is
\begin{equation}
\ket{\psi}_{\rm init} = \ket{\psi}_A\otimes\Big(\bigotimes_{i=1}^n\ket{\Phi_d}_{S_i,N_i}\Big).
\end{equation}
Hence, applying the encryption unitary to the initial state yields the encrypted state, which is given by
\begin{eqnarray}
    \ket{\psi}_{\rm enc} &=& U_{{\rm enc}_{d}}^{(n)} \otimes I_d^{\otimes n}\ket{\psi}_{\rm init} \nonumber\\
    &=&\frac{1}{d}\sum_{k,l=0}^{d-1}c_{kl}\Big(({X_d^{\scriptscriptstyle(A)}})^{k}({Z_d^{\scriptscriptstyle(A)}})^{l} \Big)\ket{\psi}_A \otimes \left(\bigotimes_{i=1}^{n} \Big(({X_d^{\scriptscriptstyle(S_i)}})^{k}({Z_d^{\scriptscriptstyle(S_i)}})^{l} \otimes I_d^{\scriptscriptstyle(N_i)}\Big)\ket{\Phi_d}\right).
\end{eqnarray}
For clarity, we omit all superscripts in the following descriptions. Furthermore, we can express the density matrix of the state as
\begin{equation}
\label{eq:jiamitai}
     \rho_{\rm enc} = \frac{1}{d^2}\sum_{k,l,m,t=0}^{d-1}c_{kl}c_{mt}^{*}X_d^{k}Z_d^{l}\ket{\psi}_A\bra{\psi}Z_d^{-t}X_d^{-m}
    \otimes \left(\bigotimes_{i=1}^{n}(X_d^{k}Z_d^{l}\otimes I_d) \ket{\Phi_d}\bra{\Phi_d}(Z_d^{-t}X_d^{-m}\otimes I_d)\right),
\end{equation}
where
\begin{equation}
    \label{eq:xishu}
    c_{kl}c_{mt}^{*}=\exp\left(-{\rm i}\pi\frac{k(k+\epsilon_d)+l(l+\epsilon_d)-m(m+\epsilon_d)-t(t+\epsilon_d)}{d}\right)
\end{equation}
and $\epsilon_d=d\bmod 2$.
\subsection{Authorization vs. informativeness}
In Ref. \cite{gianini2026encrypted}, Gianini et al. introduced the notions of authorized and unauthorized sets and classified the subsets of qubit-encrypted cloning storage registers into three types: fully informative, completely uninformative and partially informative. Building on this framework, this section presents the corresponding concepts for encrypted cloning on qudits.

The \textit{authorized sets} for the encryption-cloning scheme on the qudit are defined the following rules:

(\textbf{Auth1}) A subset of the total system is termed authorized if it includes one complete signal-noise pair and, in addition, at least one qudit from each of the remaining
$n-1$ pairs. As an illustration, the set formed by a single encrypted clone qudit together with all the noise qudits is authorized.

(\textbf{Auth2}) Equivalently, a subset is unauthorized whenever it fails to satisfy this condition; for instance, any subsystem composed of $n-1$ complete pairs is unauthorized.

(\textbf{Auth3}) Additionally, any subsystem consisting of all $n$ noise qudits alone is also unauthorized.

It is important to note that unauthorized access (i.e., the inability to recover the original qudit) does not necessarily imply a complete absence of information leakage or a maximally mixed state. The authorized subsets may, in principle, retain partial information about the input state. Prior work has already identified some unauthorized subsets that are \textit{completely uninformative}:

(\textbf{Inf1})
For $n>1$ each individual clone qudit is maximally mixed and therefore locally uninformative;

(\textbf{Inf2})
the source qudit $A$ is likewise maximally mixed after encoding;

(\textbf{Inf3}) a subset of $\{S_j, N_j\}$ lacking a complete pair is completely uninformative about $\ket{\psi}_A$.

Whether unauthorized subsets beyond those already considered are completely independent of the input state or instead exhibit residual leakage remains to be determined. The distinction between a reduced subsystem being entirely non-informative or retaining partial information is dictated by the interference structure among the generalized Pauli branches of the encoded state in Eq.(\ref{eq:jiamitai}).

Building on the work of Cear{\v{a}} \cite{cear2026clon}, which characterizes \textit{fully authorized subsets}, we aim to identify those that are \textit{completely uninformative} (i.e., independent of the input state  \(\ket{\psi}_A\)) versus those that are \textit{partially informative}. For the latter, we also provide an expression for the reduced state to quantify the leakage. Furthermore, consider the \textit{storage register} consisting of $2n$ qudits, defined as
$$
{\mathcal R}_n\equiv \{S_1,N_1,S_2,N_2,\cdots,S_n,N_n\}.
$$
For any subset $B \subseteq {\mathcal R}_{n}$, we say that
$B$ is \textit{completely uninformative} if its reduced state
\(\rho_B(\psi)\) is independent of the input state \(\ket{\psi}_A\).
In contrast, we say that $B$ exhibits leakage if
\(\rho_B(\psi)\) retains a nontrivial dependence on \(\ket{\psi}_A\), even when
$B$ is not an authorized subset and therefore does not permit full recovery of the input.

A special case of \textit{completely uninformative} state is a maximally mixed state. However, we will meet completely uninformative states which are not maximally mixed.

A  class of subsets important for our discussion are those that miss a complete signal-noise pair $\{S_i,N_i\}$ and it is easy to show that they are also \textit{completely uninformative}. This will help delimit the scope of our analysis of the reduced states.

\begin{proposition}
Any subset \(B\subset \mathcal R_n\) that misses a complete pair \(\{S_j,N_j\}\) is completely non-informative.
\end{proposition}
\begin{proof}
    First, let us consider the partial trace of system $A$.
The full encoded state $\rho_{\mathrm{enc}}^{(n)}$ consists of the qudit $A$ (which is no longer in the input state $\ket{\psi}_A$ due to the encoding) together with $n$ pairs $\{S_i,N_i\}$. To obtain the reduced state of any subset of ${\mathcal R}_n$, we trace out
$A$ from Eq.(\ref{eq:jiamitai}). Hence, the factor corresponding to $A$ is replaced by
\begin{equation}
\label{eq:traceA}
\begin{aligned}
&\tr_A\!\left(X_d^{k}Z_d^{l}\ket{\psi}_A\bra{\psi}Z_d^{-t}X_d^{-m}\right)=\bra{\psi}Z_d^{-t}X_d^{k-m}Z_d^{l}\ket{\psi}
=\omega^{-t(k-m)}\bra{\psi}X_d^{k-m}Z_d^{l-t}\ket{\psi}.
\end{aligned}
\end{equation}
The phase in Eq.~(\ref{eq:traceA}) follows from $Z_d^{-t}X_d^{k-m}=\omega^{-t(k-m)}X_d^{k-m}Z_d^{-t}$. In particular, when $l=t$, conjugating $X_d^{k-m}$ by $Z_d^t$ still produces the phase $\omega^{-t(k-m)}$; the two $Z_d$ factors cannot be cancelled across $X_d^{k-m}$. For $k=m$ and $l=t$, the scalar is equal to $1$.

Next, we trace out a full pair $\{S_j,N_j\}$ from the storage register. When we remove an entire pair $\{S_j,N_j\}$ from ${\mathcal R}_n$, we do so by tracing out both $S_j$ and $N_j$ from the factor $\Big(X_d^{k}Z_d^{l}\otimes I_d) \ket{\Phi_d}\bra{\Phi_d}(Z_d^{-t}X_d^{-m}
\otimes I_d\Big)_{\{S_jN_j\}}$ in Eq.(\ref{eq:jiamitai}).
Thus,
\begin{equation}
    \tr_{S_j N_j}\left((X_d^{k}Z_d^{l}\otimes I_d) \ket{\Phi_d}\bra{\Phi_d}(Z_d^{-t}X_d^{-m}\otimes I_d)\right)=
    \omega^{-t(k-m)}\bra{\Phi_d}(X_d^{k-m}Z_d^{l-t}\otimes I_d)\ket{\Phi_d}
    =\delta^{(d)}_{k,m}\delta^{(d)}_{l,t},
\end{equation}
where we reserve $\delta^{(d)}_{u,v}$ for the modular Kronecker delta, equal to $1$ when $u\equiv v\pmod d$ and to $0$ otherwise. Thus, by orthonormality of the generalized Bell basis, every off-diagonal term with $k\not\equiv m\pmod d$ or $l\not\equiv t\pmod d$ vanishes in Eq.~(\ref{eq:jiamitai}). After tracing out $A$ as well, the remaining diagonal coefficients are equal to $1$, and we obtain
\begin{equation}
  \rho_{\mathrm{red}}= \frac{1}{d^2} \sum_{k,l=0}^{d-1}\bigotimes_{i \neq j}\,
\Big((X_d^{k}Z_d^{l}\otimes I_d) \ket{\Phi_d}\bra{\Phi_d}(Z_d^{-l}X_d^{-k}\otimes I_d)\Big)_{S_i N_i}.
\end{equation}
This reduced state no longer depends on the initial state
$\ket{\psi}_A$; hence it is completely uninformative about it. Although the reduced state is not maximally mixed, it is \textit{completely uninformative} about $\ket{\psi}_A$.
Thus, any subset of $B\subset{\mathcal R}_n$ missing even a single full pair  $\{S_j,N_j\}$ is completely uninformative about the initial state $\ket{\psi}_A$.
\end{proof}

In Ref.\cite{gianini2026encrypted}, Gianini et al. analyzed the issue of subset size $B\subset{\mathcal R}_n$ for the case of $|B|= n$.
In this paper, we extend this study to the setting of encrypted cloning on qudits. Up to permutation symmetry, such a subset $B$ can be written as
\[
B_{n,p}=\{S_1,S_2,\dots,S_p,N_{p+1},\dots,N_{n-1}, N_n\}
\]
which contains $p\geq 0$ signal qudits and $(n-p)$ noise qudits.  Since $B_{n,p}$ is an unauthorized subset, it cannot be fully informative; we thus need to determine whether it is partially or completely uninformative. Accordingly, in the following we analyze the reduced states of all unauthorized subsets of size $|B|=n$ of the storage register ${\mathcal R}_n$.
\section{Classification of subsets of the storage register}
\label{sec:Classification}
\subsection{The low-dimensional cases}
\label{sec:low-dim}

To clearly illustrate the mechanism of data leakage, we first consider several low-dimensional examples. For convenience, we begin with $n=1$. Although the protocol is designed for $n>1$, this special case provides valuable guidance by revealing the fundamental mechanism of data leakage.

\textbf{(\RomanNum{1})\ Case $n=1$.}
For \(n=1\), the density matrix of the encoded state can be written as
\begin{equation*}
     \rho^{(1)}_{\rm enc} = \frac{1}{d^2}\sum_{k,l,m,t=0}^{d-1}c_{kl}c_{mt}^{*}\left(X_d^{k}Z_d^{l}\ket{\psi}\bra{\psi}Z_d^{-t}X_d^{-m}\right)_A
    \otimes \left((X_d^{k}Z_d^{l}\otimes I_d) \ket{\Phi_d}\bra{\Phi_d}(Z_d^{-t}X_d^{-m}\otimes I_d)\right)_{S_1N_1}.
\end{equation*}

According to Rule (\textbf{Auth1}), the set $\{S_1, N_1\}$ possesses complete informativeness. Therefore, the interesting subsets worthy of study are $\{S_1\}$ and $\{N_1\}$. First, we must apply the partial trace operation consistently to system $A$; depending on the specific subset, the partial trace operation must also be performed on either $S_1$ or $N_1$. To this end, we review some useful identities. By taking the partial trace over the quantum system $A$, we obtain the following scalar
\begin{equation*}
\tr_A\left((X_d^{k}Z_d^{l}\ket{\psi}\bra{\psi}Z_d^{-t}X_d^{-m})_A\right)=\omega^{-t(k-m)}\bra{\psi}X_d^{k-m}Z_d^{l-t}\ket{\psi},
\end{equation*}
whereas, tracing out $N_1$ or $S_1$ would yield respectively
\begin{eqnarray}
\label{eq:bellN1}
\color{red}
   &\quad& \tr_{N_1}\left[\left((X_d^{k}Z_d^{l}\otimes I_d) \ket{\Phi_d}\bra{\Phi_d}(Z_d^{-t}X_d^{-m}\otimes I_d)\right)_{S_1N_1}\right]\nonumber\\
&=&\frac{1}{d}X_d^kZ_d^{l-t}X_d^{-m}\nonumber\\
&=&\frac{1}{d}\omega^{-m(l-t)}X_d^{k-m}Z_d^{l-t},
\end{eqnarray}
and
\begin{eqnarray}
\label{eq:bellS1}
\color{red}
   &\quad& \tr_{S_1}\left[\left((X_d^{k}Z_d^{l}\otimes I_d) \ket{\Phi_d}\bra{\Phi_d}(Z_d^{-t}X_d^{-m}\otimes I_d)\right)_{S_1N_1}\right]\nonumber\\
&=&\frac{1}{d}\left(Z_d^{-t}X_d^{-m}X_d^kZ_d^l\right)^{\mathsf T}\nonumber\\
&=&\frac{1}{d}\omega^{l(m-k)}X_d^{m-k}Z_d^{l-t}.
\end{eqnarray}
Here we used the maximally entangled-state identities $\tr_N[(A\otimes I)\ket{\Phi_d}\bra{\Phi_d}(B\otimes I)]=AB/d$ and $\tr_S[(A\otimes I)\ket{\Phi_d}\bra{\Phi_d}(B\otimes I)]=(BA)^{\mathsf T}/d$, together with $X_d^{\mathsf T}=X_d^{-1}$ and $Z_d^{\mathsf T}=Z_d$. These identities fix all phases in the reduced states below.

These yield for the subset $\{S_1\}$
\begin{eqnarray}
\label{eq:rho-s1}
\color{red}
&&\rho_{S_1}^{(1)}
=
\frac{1}{d^2}\sum_{k,l,m,t=0}^{d-1}c_{kl}c_{mt}^{*}\tr_A\left((X_d^{k}Z_d^{l}\ket{\psi}\bra{\psi}Z_d^{-t}X_d^{-m})_A\right)\nonumber\\
 &&  \qquad\qquad \otimes \tr_{N_1}\left(\big((X_d^{k}Z_d^{l}\otimes I_d) \ket{\Phi_d}\bra{\Phi_d}(Z_d^{-t}X_d^{-m}\otimes I_d)\big)_{S_1N_1}\right)\nonumber\\
&&\qquad=\frac{1}{d^3}\sum_{k,l,m,t=0}^{d-1}c_{kl}c_{mt}^{*}\omega^{-t(k-m)-m(l-t)}\bra{\psi}X_d^{k-m}Z_d^{l-t}\ket{\psi}X_d^{k-m}Z_d^{l-t}
\end{eqnarray}
and for the subset $\{N_1\}$
\begin{eqnarray}
\label{eq:rho-n1}
\color{red}
&&\rho_{N_1}^{(1)}
=
\frac{1}{d^2}\sum_{k,l,m,t=0}^{d-1}c_{kl}c_{mt}^{*}\tr_A\left((X_d^{k}Z_d^{l}\ket{\psi}\bra{\psi}Z_d^{-t}X_d^{-m})_A\right)\nonumber\\
 && \qquad\qquad  \otimes \tr_{S_1}\left(\big((X_d^{k}Z_d^{l}\otimes I_d) \ket{\Phi_d}\bra{\Phi_d}(Z_d^{-t}X_d^{-m}\otimes I_d)\big)_{S_1N_1}\right)\nonumber\\
&&\qquad=\frac{1}{d^3}\sum_{k,l,m,t=0}^{d-1}c_{kl}c_{mt}^{*}\omega^{(m-k)(l+t)}\bra{\psi}X_d^{k-m}Z_d^{l-t}\ket{\psi}X_d^{m-k}Z_d^{l-t}.
\end{eqnarray}
To simplify the states, we set $a\equiv k-m\pmod d$ and $b\equiv l-t\pmod d$, so that $k\equiv m+a\pmod d$ and $l\equiv t+b\pmod d$. Here and below, $a,b\in\mathbb Z_d$, and every subscript of $c_{kl}$ is reduced modulo $d$. Hence, the states $\rho_{S_1}^{(1)}$ in Eq.(\ref{eq:rho-s1}) and
$\rho_{N_1}^{(1)}$ in Eq.(\ref{eq:rho-n1}) are changed into
\begin{equation}
\label{eq:rhos11}
\rho_{S_1}^{(1)}
=\frac{1}{d^3}\sum_{a,b=0}^{d-1}\sum_{m,t=0}^{d-1}c_{m+a,t+b}\cdot c_{mt}^{*}\omega^{-ta-mb}\bra{\psi}X_d^{a}Z_d^{b}\ket{\psi}X_d^{a}Z_d^{b}
\end{equation}
and
\begin{equation}
\label{eq:rhon11}
\rho_{N_1}^{(1)}
=\frac{1}{d^3}\sum_{a,b=0}^{d-1}\sum_{m,t=0}^{d-1}c_{m+a,t+b}\cdot c_{mt}^{*}\omega^{-a(2t+b)}\bra{\psi}X_d^{a}Z_d^{b}\ket{\psi}X_d^{-a}Z_d^{b}.
\end{equation}

Next, we need to calculate the following formulas
\begin{eqnarray}
\label{eq:xishi1}
  && \sum_{m,t=0}^{d-1}c_{m+a,t+b}\cdot c_{mt}^{*}\omega^{-ta-mb}\nonumber\\
  &=&\sum_{m,t=0}^{d-1}\exp\left(-{\rm i}\pi\frac{(m+a)(m+a+\epsilon_d)+(t+b)(t+b+\epsilon_d)-m(m+\epsilon_d)-t(t+\epsilon_d)}{d}\right)\cdot \exp\left(\frac{{\rm  i}\pi(-2ta-2mb)}{d}\right)\nonumber\\
  &=&\exp\left(\frac{-{\rm i}\pi(a^2+b^2+\epsilon_d(a+b))}{d}\right)\cdot\left(\sum_{m=0}^{d-1}\exp(-2\pi{\rm i}m\frac{a+b}{d})\right)\cdot\left(\sum_{t=0}^{d-1}\exp(-2\pi{\rm i}t\frac{a+b}{d})\right)
\end{eqnarray}
and
\begin{eqnarray}
\label{eq:xishu2}
  &&  \sum_{m,t=0}^{d-1}c_{m+a,t+b}\cdot c_{mt}^{*}\omega^{-a(2t+b)}\nonumber\\
  \color{red}
  &=& \sum_{m,t=0}^{d-1}\exp\left(-{\rm i}\pi\frac{(m+a)(m+a+\epsilon_d)+(t+b)(t+b+\epsilon_d)-m(m+\epsilon_d)-t(t+\epsilon_d)}{d}\right)\cdot \exp\left(-\frac{2{\rm i}\pi a(2t+b)}{d}\right)\nonumber\\
   &=&\exp\left(\frac{-{\rm i}\pi((a+b)^2+\epsilon_d(a+b))}{d}\right)\cdot\left(\sum_{m=0}^{d-1}\exp(-2\pi{\rm i}m\frac{a}{d})\right)\cdot\left(\sum_{t=0}^{d-1}\exp(-2\pi{\rm i}t\frac{2a+b}{d})\right).
\end{eqnarray}
Using the identity
\begin{equation}
\label{eq:xingzhi1}
    \sum_{x=0}^{d-1}e^{\frac{-2\pi{\rm i}xk}{d}}=
    \begin{cases}
        d,\ {\rm if}\ \ k=0 \bmod{d};\\
        0,\ {\rm otherwise},
    \end{cases}
\end{equation}
 Eq.(\ref{eq:xishi1}) is simplified to
 \begin{equation}
     \sum_{m,t=0}^{d-1}c_{m+a,t+b}\cdot c_{mt}^{*}\omega^{-ta-mb}
  =d^2\cdot\exp\left(\frac{-{\rm i}\pi(a^2+b^2+\epsilon_d(a+b))}{d}\right)\cdot \delta^{(d)}_{a+b,0}.
 \end{equation}
When $a+b=0\bmod d$, then we obtain that
\begin{equation}
\label{eq:xishurhos1}
    \sum_{m,t=0}^{d-1}c_{m+a,t+b}\cdot c_{mt}^{*}\omega^{-ta-mb}
  =d^2\cdot\exp\left(\frac{-2{\rm i}\pi a^2}{d}\right)=d^2\omega^{-a^2}.
\end{equation}
According to the result of Eq.(\ref{eq:xishurhos1}), the quantum state in Eq.(\ref{eq:rhos11}) can be simplified as follows
\begin{equation}
\label{eq:rhos12}
\rho_{S_1}^{(1)}
=\frac{1}{d}\sum_{a=0}^{d-1}\omega^{-a^2}\bra{\psi}X_d^{a}Z_d^{-a}\ket{\psi}X_d^{a}Z_d^{-a}.
\end{equation}
Thus the unauthorized set $\{S_1\}$ is partially informative.

Through Eq.(\ref{eq:xingzhi1}), Eq.(\ref{eq:xishu2}) can further be expressed as
\begin{equation}
    \label{eq:xishurhon1}
    \sum_{m,t=0}^{d-1}c_{m+a,t+b}\cdot c_{mt}^{*}\omega^{-a(2t+b)}
   =d^2\exp\left(\frac{-{\rm i}\pi((a+b)^2+\epsilon_d(a+b))}{d}\right)\cdot \delta^{(d)}_{a,0}\delta^{(d)}_{2a+b,0}.
\end{equation}
Based on the congruences $a\equiv0\pmod d$ and $2a+b\equiv0\pmod d$, we obtain $a=b=0$ in $\mathbb Z_d$. Furthermore, the state $\rho^{(1)}_{N_1}$ is reduced to
\[
\rho^{(1)}_{N_1}=\frac{1}{d}I_d.
\]
The unauthorized set $\{N_1\}$ is completely uninformative.

For $n=1$, the subset $\{S_1 \}$ is unauthorized and cannot fully recover the input qudit. Nevertheless, it retains a residual dependence on the input state through the quantity $\bra{\psi}X_d^{a}Z_d^{-a}\ket{\psi}$ (with $a\neq 0$). Because $X_d^{a}Z_d^{-a}$ need not be Hermitian, this complex expectation value is accessed through measurements of its Hermitian and anti-Hermitian parts (or, equivalently, by an ancilla-assisted unitary-expectation measurement), rather than by treating the Weyl operator itself as an observable. Cear{\v{a}} previously observed that for
$n=1$ the clone ${S_1}$ is not completely encrypted, and he noted that the protocol is intended to work also for $n>1$.

\textbf{(\RomanNum{2}) Case $n=2$.}
For \(n=2\), the density matrix of the encoded state is
\begin{eqnarray}
 \label{eq:2jiamitai}
    && \rho^{(2)}_{\rm enc} = \frac{1}{d^2}\sum_{k,l,m,t=0}^{d-1}c_{kl}c_{mt}^{*}\left(X_d^{k}Z_d^{l}\ket{\psi}\bra{\psi}Z_d^{-t}X_d^{-m}\right)_A
    \otimes \left((X_d^{k}Z_d^{l}\otimes I_d) \ket{\Phi_d}\bra{\Phi_d}(Z_d^{-t}X_d^{-m}\otimes I_d)\right)_{S_1N_1}\nonumber\\
    &&\qquad\quad\otimes \left((X_d^{k}Z_d^{l}\otimes I_d) \ket{\Phi_d}\bra{\Phi_d}(Z_d^{-t}X_d^{-m}\otimes I_d)\right)_{S_2N_2}.
\end{eqnarray}
For this scenario, we consider the unauthorized sets $\{S_1,N_2\}$, $\{S_2,N_1\}$ and $\{S_1,S_2\}$. Due to symmetry, $\{S_1,N_2\}$ and $\{S_2,N_1\}$ can be considered equivalent.

The reduced state for the subset $\{S_1,N_2\}$ is obtained by tracing out $A$ plus one noise qudit from one pair and one signal qudit from the other.
\begin{eqnarray}
\label{eq:rho-s1n2}
\color{red}
&&\rho_{S_1N_2}^{(2)}
=
\frac{1}{d^2}\sum_{k,l,m,t=0}^{d-1}c_{kl}c_{mt}^{*}\tr_A\left((X_d^{k}Z_d^{l}\ket{\psi}\bra{\psi}Z_d^{-t}X_d^{-m})_A\right)\otimes \tr_{N_1}\left(\big((X_d^{k}Z_d^{l}\otimes I_d) \ket{\Phi_d}\bra{\Phi_d}(Z_d^{-t}X_d^{-m}\otimes I_d)\big)_{S_1N_1}\right)\nonumber\\
 && \qquad\qquad\qquad  \otimes \tr_{S_2}\left(\big((X_d^{k}Z_d^{l}\otimes I_d) \ket{\Phi_d}\bra{\Phi_d}(Z_d^{-t}X_d^{-m}\otimes I_d)\big)_{S_2N_2}\right)\\
&&\qquad=\frac{1}{d^4}\sum_{k,l,m,t=0}^{d-1}c_{kl}c_{mt}^{*}\omega^{-t(k-m)-m(l-t)+l(m-k)}\bra{\psi}X_d^{k-m}Z_d^{l-t}\ket{\psi}X_d^{k-m}Z_d^{l-t}\otimes X_d^{m-k}Z_d^{l-t}.\nonumber
\end{eqnarray}

The two-signal qudit subset \(\{S_1,S_2\}\) is obtained tracing out $A$ and both the noise qudits $N_1,N_2$.
\begin{eqnarray}
\label{eq:rho-s1s2}
\color{red}
&&\rho_{S_1S_2}^{(2)}
=
\frac{1}{d^2}\sum_{k,l,m,t=0}^{d-1}c_{kl}c_{mt}^{*}\tr_A\left((X_d^{k}Z_d^{l}\ket{\psi}\bra{\psi}Z_d^{-t}X_d^{-m})_A\right)\otimes \tr_{N_1}\left(\big((X_d^{k}Z_d^{l}\otimes I_d) \ket{\Phi_d}\bra{\Phi_d}(Z_d^{-t}X_d^{-m}\otimes I_d)\big)_{S_1N_1}\right)\nonumber\\
 && \qquad\qquad  \otimes \tr_{N_2}\left(\big((X_d^{k}Z_d^{l}\otimes I_d) \ket{\Phi_d}\bra{\Phi_d}(Z_d^{-t}X_d^{-m}\otimes I_d)\big)_{S_2N_2}\right)\\
&&\qquad=\frac{1}{d^4}\sum_{k,l,m,t=0}^{d-1}c_{kl}c_{mt}^{*}\omega^{-t(k-m)-2m(l-t)}\bra{\psi}X_d^{k-m}Z_d^{l-t}\ket{\psi}X_d^{k-m}Z_d^{l-t}\otimes X_d^{k-m}Z_d^{l-t}.\nonumber
\end{eqnarray}
Using $a\equiv k-m\pmod d$ and $b\equiv l-t\pmod d$, the coefficient sums become
\begin{equation}
    \label{eq:xishurhos1n2}
    \sum_{m,t=0}^{d-1}c_{m+a,t+b}c_{mt}^{*}\omega^{-ta-mb-a(t+b)}
   =d^2\exp\left(\frac{-{\rm i}\pi((a+b)(\epsilon_d+a+b))}{d}\right)\cdot \delta^{(d)}_{a+b,0}\delta^{(d)}_{2a+b,0},
\end{equation}
and
\begin{equation}
    \label{eq:xishurhos1s2}
    \sum_{m,t=0}^{d-1}c_{m+a,t+b}c_{mt}^{*}\omega^{-ta-2mb}
   =d^2 \delta^{(d)}_{a+2b,0}\delta^{(d)}_{a+b,0},
\end{equation}
Based on the congruence systems
\begin{equation*}
    \begin{cases}
        a+b=0\bmod d\\
        2a+b=0\bmod d
    \end{cases}
\end{equation*}
or
\begin{equation*}
    \begin{cases}
        a+b=0\bmod d\\
        a+2b=0\bmod d,
    \end{cases}
\end{equation*}
we can solve for $a=0$ and $b=0$. Hence, both the examined reduced states are uninformative
\[
\rho_{S_1N_2}=\frac{1}{d^2}\big(I_d\otimes I_d\big)\]
and
\[\rho_{S_1S_2}=\frac{1}{d^2}\big(I_d\otimes I_d\big).
\]
Therefore, the unauthorized sets $\{S_1,N_2\}$ and $\{S_1,S_2\}$ are completely uninformative.

\textbf{(\RomanNum{3}) Case $n=3$.}
For $n=3$ the aligned unauthorized sets are $\{N_1,N_2,N_3\}$, $\{S_1,S_2,S_3\}$, $\{S_1,S_2,N_3\}$, $\{S_1,N_2,S_3\}$, $\{N_1,S_2,S_3\}$, $\{S_1,N_2,N_3\}$, $\{S_2,N_1,N_3\}$ and $\{S_3,N_1,N_2\}$. The all-noise subset ($p=0$) is maximally mixed by the same calculation as the general case below. It is therefore sufficient to analyze the representative sets $\{S_1,S_2,N_3\}$, $\{S_1,N_2,N_3\}$ and $\{S_1,S_2,S_3\}$; the remaining cases follow by permutation symmetry.

Tracing out the $A$ and the complementary qudits of the storage register, the reduced states can be written as follows,
\begin{eqnarray}
\label{eq:rho-s1s2n3}
&&\rho_{S_1S_2N_3}^{(3)}
=\frac{1}{d^5}\sum_{k,l,m,t=0}^{d-1}c_{kl}c_{mt}^{*}\omega^{-t(k-m)-2m(l-t)-l(k-m)}\bra{\psi}X_d^{k-m}Z_d^{l-t}\ket{\psi}X_d^{k-m}Z_d^{l-t}\otimes X_d^{k-m}Z_d^{l-t}\otimes X_d^{m-k}Z_d^{l-t}.\nonumber\\
\end{eqnarray}
\begin{eqnarray}
\label{eq:rho-s1n2n3}
&&\rho_{S_1N_2N_3}^{(3)}
=\frac{1}{d^5}\sum_{k,l,m,t=0}^{d-1}c_{kl}c_{mt}^{*}\omega^{-t(k-m)-m(l-t)-2l(k-m)}\bra{\psi}X_d^{k-m}Z_d^{l-t}\ket{\psi}X_d^{k-m}Z_d^{l-t}\otimes X_d^{m-k}Z_d^{l-t}\otimes X_d^{m-k}Z_d^{l-t}.\nonumber\\
\end{eqnarray}
\begin{eqnarray}
\label{eq:rho-s1s2s3}
&&\rho_{S_1S_2S_3}^{(3)}
=\frac{1}{d^5}\sum_{k,l,m,t=0}^{d-1}c_{kl}c_{mt}^{*}\omega^{-t(k-m)-3m(l-t)}\bra{\psi}X_d^{k-m}Z_d^{l-t}\ket{\psi}X_d^{k-m}Z_d^{l-t}\otimes X_d^{k-m}Z_d^{l-t}\otimes X_d^{k-m}Z_d^{l-t}.\nonumber\\
\end{eqnarray}

Using $a\equiv k-m\pmod d$ and $b\equiv l-t\pmod d$, the coefficients in Eq.(\ref{eq:rho-s1s2n3}), Eq.(\ref{eq:rho-s1n2n3}) and Eq.(\ref{eq:rho-s1s2s3}) are simplified as follows:
\begin{eqnarray}
    \label{eq:xishurhos1s2n3}
    \color{red}&&\qquad\sum_{m,t=0}^{d-1}c_{m+a,t+b}c_{mt}^{*}\omega^{-ta-2mb-a(t+b)}\nonumber\\
   &&=d^2\exp\left(\frac{-{\rm i}\pi(a^2+b^2+2ab+\epsilon_d(a+b))}{d}\right)\cdot \delta^{(d)}_{a+2b,0}\delta^{(d)}_{2a+b,0},
\end{eqnarray}
\begin{eqnarray}
    \label{eq:xishurhos1n2n3}
    \color{red}&&\qquad\sum_{m,t=0}^{d-1}c_{m+a,t+b}c_{mt}^{*} \omega^{-ta-mb-2a(t+b)}\nonumber\\
   &&=d^2\exp\left(\frac{-{\rm i}\pi(a^2+b^2+4ab+\epsilon_d(a+b))}{d}\right)\cdot \delta^{(d)}_{a+b,0}\delta^{(d)}_{3a+b,0},
\end{eqnarray}
\begin{eqnarray}
    \label{eq:xishurhos1s2s3}
    \color{red}&&\qquad\sum_{m,t=0}^{d-1}c_{m+a,t+b}c_{mt}^{*}\omega^{-ta-3mb}\nonumber\\
   &&=d^2\exp\left(\frac{-{\rm i}\pi(a^2+b^2+\epsilon_d(a+b))}{d}\right)\cdot \delta^{(d)}_{a+3b,0}\delta^{(d)}_{a+b,0}.
\end{eqnarray}
By Eq.(\ref{eq:xishurhos1s2n3}), for the congruence systems
    $\begin{cases}
        a+2b=0\bmod d\\
        2a+b=0\bmod d
    \end{cases}$,
we can obtain that
\begin{equation}
    \begin{cases}
        (a,b)=(0,0),& {\rm if\ }\ \gcd(d,3)=1,\\
        (a,b)\in\left\{(0,0),\left(\frac{d}{3},\frac{d}{3}\right),\left(\frac{2d}{3},\frac{2d}{3}\right)\right\},& {\rm if\ }\ \gcd(d,3)=3,
    \end{cases}
\end{equation}
where $\gcd(u,v)$ denotes the greatest common divisor of $u$ and $v$.
Hence, both the examined reduced state is
\begin{equation}
    \rho_{S_1S_2N_3}^{(3)}= \begin{cases}
        \frac{I_d^{\otimes 3}}{d^3},& {\rm if\ \  gcd}(d,3)=1,\\
        \\
        \frac{1}{d^3}\Big(I_d^{\otimes 3}+x_1(X_d^{\frac{d}{3}}Z_d^{\frac{d}{3}})^{\otimes 2}\otimes(X_d^{\frac{2d}{3}}Z_d^{\frac{d}{3}})+y_1(X_d^{\frac{2d}{3}}Z_d^{\frac{2d}{3}})^{\otimes 2}\otimes(X_d^{\frac{d}{3}}Z_d^{\frac{2d}{3}})\Big),
        & {\rm if\ \ gcd}(d,3)=3,
    \end{cases}
\end{equation}
where $x_1=\exp\left(-{\rm i}\pi(\frac{4d}{9}+\frac{2\epsilon_d}{3})\right)
\bra{\psi}X_d^{\frac{d}{3}}Z_d^{\frac{d}{3}}\ket{\psi}$ and $y_1=\exp\left(-{\rm i}\pi(\frac{16d}{9}+\frac{4\epsilon_d}{3})\right)
\bra{\psi}X_d^{\frac{2d}{3}}Z_d^{\frac{2d}{3}}\ket{\psi}$.

By Eq.(\ref{eq:xishurhos1n2n3}) and Eq.(\ref{eq:xishurhos1s2s3}), for the congruence systems
\begin{equation}
   \begin{cases}
        a+b=0\bmod d\\
        3a+b=0\bmod d
    \end{cases}
\end{equation}
    or
\begin{equation}
   \begin{cases}
        a+b=0\bmod d\\
        a+3b=0\bmod d
    \end{cases}
\end{equation}
we can solve that
\begin{equation}
    \begin{cases}
        (a,b)=(0,0),& {\rm if\ }\ \gcd(d,2)=1,\\
        (a,b)\in\left\{(0,0),\left(\frac{d}{2},\frac{d}{2}\right)\right\},& {\rm if\ }\ \gcd(d,2)=2.
    \end{cases}
\end{equation}
Hence, both the examined reduced states are
\begin{equation}
    \rho_{S_1N_2N_3}^{(3)}=\rho_{S_1S_2S_3}^{(3)}= \begin{cases}
        \frac{I_d^{\otimes 3}}{d^3},& {\rm if\ \  gcd}(d,2)=1,\\
        \\
        \frac{1}{d^3}\Big(I_d^{\otimes 3}+{\rm i}^d\bra{\psi}X_d^{\frac{d}{2}}Z_d^{\frac{d}{2}}\ket{\psi}(X_d^{\frac{d}{2}}Z_d^{\frac{d}{2}})^{\otimes 3}\Big),
        & {\rm if\ \ gcd}(d,2)=2.
    \end{cases}
\end{equation}

For $n=3$, the $p=0$ all-noise subset is always maximally mixed. The $p=1$ and $p=3$ classes are maximally mixed when $\gcd(d,2)=1$ and partially informative when $\gcd(d,2)=2$, whereas the $p=2$ class is maximally mixed when $\gcd(d,3)=1$ and partially informative when $\gcd(d,3)=3$.
\subsection{The general case}
\label{sec:general}
The examples of Section~\ref{sec:low-dim} suggest that leakage in encrypted cloning is not generic, but highly structured.
For the general case of $n$, we find that it depends on the solution of the congruence system.

By the reduction established in Section~\ref{sec:Preliminary}, the only unauthorized subsets that remain to be classified are, up to permutation of the pairs, the aligned subsets
\[
B_{n,p}=\{S_1,\dots,S_p,N_{p+1},\dots,N_n\},
\qquad 0\le p\le n.
\]
From now on, we will denote by $p$ the number of signal qudits, by $(n-p)$ the number of noise qudits.

For these subsets, we compute the partial trace of the quantum system \(A\) and the complementary quantum systems in each pair of qudits, ultimately obtaining the reduced state.
\begin{eqnarray}
\label{eq:rho-s1s2np}
&&\rho_{B_{n,p}}
=\frac{1}{d^{2+n}}\sum_{k,l,m,t=0}^{d-1}c_{kl}c_{mt}^{*}\omega^{-t(k-m)-pm(l-t)-l(k-m)(n-p)}\bra{\psi}X_d^{k-m}Z_d^{l-t}\ket{\psi} (X_d^{k-m}Z_d^{l-t})^{\otimes p}\nonumber\\
&&\qquad\qquad\qquad\qquad\qquad\otimes (X_d^{m-k}Z_d^{l-t})^{\otimes (n-p)}.
\end{eqnarray}

Using $a\equiv k-m\pmod d$ and $b\equiv l-t\pmod d$, the coefficient in Eq.(\ref{eq:rho-s1s2np}) is simplified as follows:
\begin{equation}
\label{eq:xishurhos1s2np}
\begin{aligned}
\Gamma_{a,b}
&:=\sum_{m,t=0}^{d-1}c_{m+a,t+b}c_{mt}^{*}
\omega^{-ta-pmb-a(t+b)(n-p)}\\
&=\exp\!\left[-\frac{{\rm i}\pi}{d}
\bigl(a^2+b^2+2ab(n-p)+\epsilon_d(a+b)\bigr)\right]\\
&\quad\times\left(\sum_{m=0}^{d-1}
\exp\!\left[-\frac{2{\rm i}\pi m(a+pb)}{d}\right]\right)
\left(\sum_{t=0}^{d-1}
\exp\!\left[-\frac{2{\rm i}\pi t((n-p+1)a+b)}{d}\right]\right)\\
&=d^2\exp\!\left[-\frac{{\rm i}\pi}{d}
\bigl(a^2+b^2+2ab(n-p)+\epsilon_d(a+b)\bigr)\right]
\delta^{(d)}_{a+pb,0}\delta^{(d)}_{(n-p+1)a+b,0}.
\end{aligned}
\end{equation}
Furthermore, we analyze the congruence system
\begin{equation}
\label{eq:tyfczu}
    \begin{cases}
        a+pb\equiv0 \pmod d,\\
        (n-p+1)a+b\equiv0 \pmod d,
    \end{cases}
\end{equation}
where $p\in\{0,1,\ldots,n\}$, $n\geq 2$, and $d\geq 2$. The solutions in $\mathbb Z_d^2$ are
\begin{equation}
    \begin{cases}
        a_v\equiv-pv\,\dfrac{d}{g_1}\pmod d,\\[2pt]
        b_v\equiv v\,\dfrac{d}{g_1}\pmod d,
    \end{cases}
\end{equation}
where $v=0,1,\ldots,g_1-1$, $g_1=\gcd\!\left(d,p(n-p+1)-1\right)$, and $a_v,b_v$ denote the representatives in $\{0,1,\ldots,d-1\}$.

Hence, we obtain the complete reduced-state classification for all aligned unauthorized subsets.
\begin{equation}
\label{eq:fenlei1}
\rho_{B_{n,p}}
=\frac{1}{d^n}
\begin{cases}
 I_d^{\otimes n},
 & \text{if }\gcd\!\left(d,p(n-p+1)-1\right)=1,\\[1.2ex]
 \displaystyle
 \sum_{v=0}^{g_1-1}\beta_v
 \bra{\psi}X_d^{a_v}Z_d^{b_v}\ket{\psi}
 (X_d^{a_v}Z_d^{b_v})^{\otimes p}
 \otimes (X_d^{-a_v}Z_d^{b_v})^{\otimes(n-p)},
 & \text{if }\gcd\!\left(d,p(n-p+1)-1\right)>1,
\end{cases}
\end{equation}
where
$\beta_v:=\exp\!\left[-\frac{{\rm i}\pi}{d}\bigl(a_v^2+b_v^2+2a_vb_v(n-p)+\epsilon_d(a_v+b_v)\bigr)\right]$.
This abbreviation makes the operator and phase conventions in the general formula explicit.
Specifically, when $d=2$, $g_1=\gcd\!\left(2,p(n-p+1)-1\right)\in\{1,2\}$. If $g_1=2$, then both $n$ and $p$ are odd. Therefore,
Eq.(\ref{eq:fenlei1}) can be simplified to the following expression,
\begin{equation}
\rho_{B_{n,p}}
=
\frac{1}{2^n}
\begin{cases}
I_2^{\otimes n},
& n \text{ even or } p \text{ even},\\[1ex]
I_2^{\otimes n}-\bra{\psi}X_2Z_2\ket{\psi} (X_2Z_2)^{\otimes n},
& n \text{ odd and } p \text{ odd},
\end{cases}
\end{equation}
which is consistent with the result presented in Ref.\cite{gianini2026encrypted}.
\section{Improvement of information leakage for \texorpdfstring{$B_{n,p}$}{B(n,p)}}
\label{sec:Improvement}
In Sec.~\ref{sec:general}, we analyze the information leakage scenario for the set $B_{n,p}$ and derive the reduced density matrix $\rho_{B_{n,p}}$ given in Eq.~\eqref{eq:fenlei1}. From this expression, the leakage arises from the generalized Pauli operators associated with the nonzero solutions of the congruence system in Eq.~\eqref{eq:tyfczu}. The leakage is completely eliminated if and only if the condition $\gcd\!\left(d,p(n-p+1)-1\right)=1$ is satisfied; otherwise, the residual terms persist as a linear combination of expectation values $\langle\psi|X_d^{a_v}Z_d^{b_v}|\psi\rangle$.

To mitigate the information leakage, we adopt the standard quantum twirling technique. Instead of modifying the encoding unitary $U_{\text{enc}_d}^{(n)}$ (which would break the authorized recovery condition), we apply a random unitary $V_{r,s}$ from the full generalized Pauli group prior to the encoding. This operation does not destroy the entanglement structure but uniformly randomizes the input state's representation in the logical basis. As we will prove below, the inverse operation $V^\dagger_{r,s}$ is applied after the quantum decryption, ensuring that the authorized set's recovery remains deterministic. Crucially, the success of decryption does not require $V_{r,s}$ to commute with $U_{\text{enc}_d}^{(n)}$; it only requires the unitary nature of the encryption map. By virtue of the average orthogonality property of the group, this procedure makes every unauthorized joint quantum--classical state independent of the input. Maximal mixing is asserted only for aligned unauthorized subsets containing one qudit from each pair. The specific procedure is shown in Algorithm \ref{ALG:QECshamir}.

\begin{algorithm}
  \caption{Quantum encryption cloning for authorized sets using the Shamir threshold scheme}
  \label{ALG:QECshamir}

  \SetKwInOut{Input}{Input}
  \SetKwInOut{Output}{Output}

  \Input{\begin{enumerate}
    \renewcommand{\labelenumi}{(\theenumi)}
    \item An unknown pure state $\ket{\psi}$;
    \item The number of clone pairs $n> 1$;
    \item Participant set $\mathcal{P}=\{P_1,\dots,P_{2n}\}$.
\end{enumerate}
  }
  \Output{
    \begin{enumerate}
    \renewcommand{\labelenumi}{(\theenumi)}
    \item Any authorized subset $\mathcal{A}\subseteq\mathcal{P}$ containing a complete participant pair $\{P_j,P_{n+j}\}$ and at least one participant from every other pair recovers $\ket{\psi}$ perfectly;
    \item Any unauthorized subset obtains a joint quantum--classical state that is independent of the input state. In particular, every aligned unauthorized subset of size $n$ has the maximally mixed quantum state $I_d^{\otimes n}/d^n$.
\end{enumerate}}

\vspace{0.2cm}
\tcc{\textbf{Phase 1 -- Sharing (executed by the dealer)}}
\medskip
  \BlankLine
Choose uniformly random two elements $r,s \in \mathbb{Z}_d\subset \mathbb{F}_q$, where $\mathbb{F}_q$ is a prime field and $q> \max\{d,2n\}$\;
Prepare the twisted state $\ket{\psi_{r,s}} = V_{r,s}\ket{\psi}$, where $V_{r,s}=X_d^r Z_d^s$\;
Apply the qudit encryption unitary $U_{\text{enc}_d}^{(n)}$ to the joint system consisting of qudit $A$ (holding $\ket{\psi_{r,s}}$) and signal qudits $S_1,\dots,S_n$; noise qudits $N_i$ are untouched\;
Obtain the encoded state
\(
\ket{\Psi_{\text{enc}}} = U_{{\text{enc}}_d}^{(n)}\Bigl(\ket{\psi_{r,s}}_A \otimes \bigotimes_{i=1}^n \ket{\Phi_d}_{S_iN_i}\Bigr)
\)\;
Split $r$ into shares $\{r_i\}_{i=1}^{2n}$ using Shamir's $(n+1,2n)$-threshold scheme over $\mathbb{F}_q$\;
Independently split $s$ into shares $\{s_i\}_{i=1}^{2n}$ using the same scheme\;
Securely transmit pair $(r_i,s_i)$ to participant $P_i$ for each $i$\;
Assign quantum shares: give $S_i$ to $P_i$ and $N_i$ to $P_{n+i}$ ($i=1,\ldots,n$). The dealer retains nothing\;
  \medskip
\tcc{\textbf{Phase 2 -- Reconstruction (by authorized subset $\mathcal{A}$)}}

  \BlankLine
Authorization check: if $|\mathcal{A}| < n+1$, or $\nexists j$ such that $\{P_j,P_{n+j}\}\subseteq\mathcal{A}$, or $\exists i$ such that $\mathcal{A}\cap\{P_i,P_{n+i}\}=\varnothing$, then abort (cannot recover)\;
Participants in $\mathcal{A}$ exchange their shares $r_i$ (over an authenticated classical channel) and reconstruct $r$ via Shamir interpolation; repeat for $s_i$ to recover $s$\;
Using the authorized-subset decoder for qudit encrypted cloning, which acts only on the shares in $\mathcal{A}$ and requires both a complete pair and representation of every pair, the participants recover $\ket{\psi_{r,s}}$\;
Apply the inverse twist: $V_{r,s}^\dagger \ket{\psi_{r,s}} = Z_d^{-s} X_d^{-r} \ket{\psi_{r,s}} = \ket{\psi}$\;
Output the recovered secret state $\ket{\psi}$\;
\end{algorithm}

Next, we provide an analysis of the correctness of the authorized set in restoring the secret with probability 1.
\begin{theorem}[Correctness]
\label{thm:Correctness}
In Algorithm \ref{ALG:QECshamir}, any authorized participant subset $\mathcal{A} \subseteq \mathcal{P}=\{P_1,\ldots,P_{2n}\}$ that contains at least one complete participant pair $\{P_j,P_{n+j}\}$ and intersects every remaining pair $\{P_i,P_{n+i}\}$ (and hence satisfies $|\mathcal{A}|\ge n+1$) can reconstruct the original secret state $\ket{\psi}$ with unit probability.
\end{theorem}
\begin{proof}
The proof is structured into three consecutive steps: recovery of the classical shares, decryption of the quantum state, and removal of the applied randomisation.

\textbf{Step 1: Reconstruction of classical shares.}
In the distribution phase, the dealer randomly and uniformly chooses two elements $r,s\in\mathbb{Z}_d\subset \mathbb{F}_q$ and encodes them into $2n$ classical shares $(r_i,s_i)$ using two independent instances of Shamir's $(n+1,2n)$-threshold secret-sharing scheme over $\mathbb{F}_q$ \cite{Shamir}. Each share $(r_i,s_i)$ is securely transmitted to participant $P_i$.
For any authorized set $\mathcal{A}$ with $|\mathcal{A}|\geq n+1$, the completeness property of Shamir's scheme ensures that the original secrets $r$ and $s$ can be uniquely recovered via Lagrange interpolation.

For example, we present a detailed construction of the classical shares encoding the secret $r \in \mathbb{Z}_d$, following the same distribution phase described above. The construction is based on Shamir's $(n+1, 2n)$-threshold secret-sharing scheme over $\mathbb{F}_q$.

\textbf{1) Parameter setup.} Let $\mathcal{P} = \{P_1, \cdots, P_{2n}\}$ be the set of $2n$ participants. Fix $2n$ distinct, publicly known evaluation points $x_1, \cdots, x_{2n} \in \mathbb{F}_q$, with $x_i \neq 0$ for convenience. Here, the condition $q> 2n$ ensures that it is possible to select $2n$ distinct non-zero evaluation points; the secret is uniformly selected only from the embedded set $\mathbb{Z}_d \subset \mathbb{F}_q$.

\textbf{2) Polynomial generation.} To share the secret $r \in \mathbb{Z}_d\subset \mathbb{F}_q$, the dealer randomly and uniformly chooses $n$ auxiliary coefficients $a_1, \cdots, a_n \in \mathbb{F}_q$ independently. These coefficients define the polynomial
\[
f(x) = r + a_1 x + a_2 x^2 + \cdots + a_n x^n \pmod q,
\]
of degree at most $n$. The constant term of $f$ is precisely the secret $r$.

\textbf{3) Share computation.} For each participant $P_i$ ($i = 1, \cdots, 2n$), the dealer evaluates $f$ at the point $x_i$ and assigns the resulting value as the share
\[
r_i = f(x_i) \in \mathbb{F}_q.
\]
The share $r_i$ is then securely transmitted to $P_i$ over an authenticated private channel.

\textbf{4) Reconstruction property.} By the standard Lagrange interpolation formula, any subset $\mathcal{A} \subseteq \mathcal{P}$ of size $|\mathcal{A}| \ge n+1$ can uniquely recover the polynomial $f$ from its evaluations at the corresponding points $\{x_i : P_i \in \mathcal{A}\}$. Indeed, the interpolating polynomial is given by
\[
f(x) = \sum_{i \in \mathcal{A}} r_i \cdot L_i(x), \qquad
L_i(x) = \prod_{\substack{j \in \mathcal{A} \\ j \neq i}} \frac{x - x_j}{x_i - x_j},
\]
with all arithmetic performed in $\mathbb{F}_q$. The secret $r$ is then obtained as the constant term $f(0) = r$.

\textbf{5) Security guarantee.} Suppose that a coalition holds $t\leq n$ shares. For every candidate secret $r_0\in\mathbb{Z}_d$, the number of degree-at-most-$n$ polynomials over $\mathbb{F}_q$ that are consistent with the observed shares and have constant term $r_0$ is $q^{n-t}$, rather than one. Since this number is independent of $r_0$, and since the remaining polynomial coefficients are sampled uniformly, the distribution of the coalition's shares is independent of $r$. The same argument applies independently to $s$.

\textbf{Step 2: Quantum decryption via encrypted cloning.}
Let $\mathcal E_n$ denote the channel from the input qudit to the storage register obtained by applying $U_{{\rm enc}_d}^{(n)}$ to the input and the $n$ Bell pairs and then tracing out system $A$. The recovery result for qudit encrypted cloning states that, whenever $\mathcal A$ contains a complete pair and at least one qudit from every remaining pair, there is a decoding channel $\mathcal D_{\mathcal A}$ whose input consists only of the shares held by $\mathcal A$ and which satisfies
\begin{equation}
\mathcal D_{\mathcal A}\!\left[
\tr_{\mathcal R_n\setminus\mathcal A}
\mathcal E_n\!\left(\ket{\xi}\bra{\xi}\right)
\right]=\ket{\xi}\bra{\xi}
\quad\text{for every }\ket{\xi}.
\end{equation}
One may select a complete pair in $\mathcal A$ together with one held qudit from each remaining pair and discard any additional shares before applying this decoder. Taking $\ket{\xi}=\ket{\psi_{r,s}}$ therefore recovers $\ket{\psi_{r,s}}$ without access to any share outside $\mathcal A$; no inversion of the global encoding unitary on unavailable systems is invoked.

\textbf{Step 3: Undoing the randomisation.}
Once the participants have recovered both the classical pair $(r,s)$ (from Step~1) and the quantum state $\ket{\psi_{r,s}}$ (from Step~2), they locally apply the inverse unitary $V_{r,s}^\dagger = Z_d^{-s} X_d^{-r}$ to the recovered state. Since $V_{r,s}$ is unitary, we have
\begin{equation}
V_{r,s}^\dagger \, \ket{\psi_{r,s}} = V_{r,s}^\dagger \, V_{r,s} \, \ket{\psi} = I_d\ket{\psi} = \ket{\psi}.
\end{equation}
This identity holds exactly as an equality of pure states, with no approximation or probability involved.

Thus, an authorized set $\mathcal{A}$ can reconstruct the original secret $\ket{\psi}$ with certainty. This completes the proof of completeness.
\end{proof}

Furthermore, we now establish the perfect secrecy of the proposed scheme in Algorithm \ref{ALG:QECshamir}: any unauthorized subset ${B}\subseteq \mathcal{R}_n$ obtains a joint quantum--classical state that is completely independent of the input state $|\psi\rangle$. The proof relies crucially on the average orthogonality of the generalized Pauli group, which we first give a lemma.

\begin{lemma}
\label{lem:weyl}
Let $d\ge2$ be an integer. Define the generalized Pauli operators
\[
V_{r,s}=X_d^r Z_d^s,\quad r,s\in\mathbb Z_d\subset \mathbb{F}_q,
\]
where $X_d$ and $Z_d$ satisfy $X_d Z_d=\omega^{-1}Z_d X_d$ with $\omega=e^{\frac{2\pi {\rm i}}{d}}$. For any generalized Pauli operator $W_{a,b}=X_d^aZ_d^b$, one has
\begin{equation}
\label{eq:lemma1}
 \frac{1}{d^2}\sum_{r,s\in\mathbb Z_d} V_{r,s}^\dagger\, W_{a,b}\, V_{r,s}
=
\begin{cases}
I_d, & (a,b)=(0,0),\\[4pt]
0, & (a,b)\neq(0,0).
\end{cases}
\end{equation}

\end{lemma}

\begin{proof}
For $(a,b)=(0,0)$ the claim is trivial. Assume $W_{a,b}\neq I$. The left-hand side of Eq.(\ref{eq:lemma1}) commutes with every $V_{u,v}$.
Define
\[
A \equiv \frac{1}{d^2} \sum_{r,s\in\mathbb Z_d} V_{r,s}^\dagger W_{a,b} V_{r,s}.
\]
For arbitrary \( u,v\in\mathbb Z_d \), using the composition property of generalized Pauli operators, there exists a phase factor \( \eta_{r,s;u,v} \) with \( |\eta|=1 \) such that
\[
V_{r,s}V_{u,v} = \eta_{r,s;u,v} V_{r+u,s+v}.
\]
Consequently,
\[
\begin{aligned}
V_{u,v}^\dagger A V_{u,v}
&= \frac{1}{d^2} \sum_{r,s} V_{u,v}^\dagger V_{r,s}^\dagger W_{a,b} V_{r,s}V_{u,v} \\
&= \frac{1}{d^2} \sum_{r,s} (V_{r,s}V_{u,v})^\dagger W_{a,b} (V_{r,s}V_{u,v}) \\
&= \frac{1}{d^2} \sum_{r,s} V_{r+u,s+v}^\dagger W_{a,b} V_{r+u,s+v}.
\end{aligned}
\]
Relabeling the summation indices as \( r'\equiv r+u \) and \( s'\equiv s+v \), which merely permutes the finite group \( \mathbb Z_d^2 \), yields
\[
V_{u,v}^\dagger A V_{u,v} = A.
\]
Thus \( [A, V_{u,v}]=0 \) for all \( u,v\in\mathbb Z_d \).

The set \( \{V_{r,s}: r,s\in\mathbb Z_d\} \) spans the entire matrix algebra \( M_d(\mathbb C) \) and forms an irreducible representation. Therefore, by Schur's lemma, any operator commuting with all \( V_{u,v} \) must be a scalar multiple of the identity. Hence we may write
\[
A = \lambda I_d
\]
for some \( \lambda\in\mathbb C \).

Taking the trace of both sides gives

\[
d\lambda= \tr(A)= \frac{1}{d^2} \sum_{r,s} \tr\left(V_{r,s}^\dagger W_{a,b} V_{r,s}\right)= \frac{1}{d^2} \sum_{r,s} \tr(W_{a,b})= \tr(W_{a,b}),
\]
where we used the cyclicity of the trace and the unitarity of \( V_{r,s} \).

For generalized Pauli operators, the trace is known to be
\[
\tr(X_d^a Z_d^b) = d\, \delta^{(d)}_{a,0}\delta^{(d)}_{b,0}.
\]
Substituting this into the previous result, we obtain
\[
\lambda = \delta^{(d)}_{a,0}\delta^{(d)}_{b,0}.
\]
Therefore,
\[
A = \delta^{(d)}_{a,0}\delta^{(d)}_{b,0}\, I_d,
\]
which is exactly equivalent to the desired statement in Eq.~(\ref{eq:lemma1}).
\end{proof}

\begin{theorem}[Perfect secrecy]
\label{thm:security}
Let $B\subseteq\mathcal R_n=\{S_1,N_1,\ldots,S_n,N_n\}$ be any subset that is unauthorized according to \textbf{Auth1}, and let $C_B$ denote all classical shares held by the corresponding participants. Then there exists a state $\sigma_{BC_B}$, independent of the input, such that the joint output channel of Algorithm~\ref{ALG:QECshamir} satisfies
\[
\mathcal E_{BC_B}(\rho_A)=\tr(\rho_A)\,\sigma_{BC_B}.
\]
If, in addition, $B$ contains exactly one qudit from each signal--noise pair, then its quantum marginal is maximally mixed:
\[
\tr_{C_B}(\sigma_{BC_B})=\frac{I_d^{\otimes n}}{d^n}.
\]
\end{theorem}

\begin{proof}
We consider two mutually exclusive cases.

\emph{Case 1.} $B$ misses at least one complete signal--noise pair $\{S_j,N_j\}$, i.e., $B\cap\{S_j,N_j\}=\emptyset$. Proposition~1 applies before any assumption about $|B|$: for every single-qudit input density operator $\xi_A$,
\[
\tr_{\mathcal R_n\setminus B}(\mathcal E_n(\xi_A))=\sigma_B,
\]
where the fixed state $\sigma_B$ is independent of $\xi$. In particular, for every fixed key $(r,s)$,
\[
\tr_{\mathcal R_n\setminus B}\left(\mathcal E_n\!\left(
V_{r,s}\ket{\psi}\bra{\psi}V_{r,s}^{\dagger}\right)\right)=\sigma_B.
\]
A missing-pair set can contain as many as $2n-2$ shares and may therefore reconstruct the Shamir keys. This does not reveal the input: conditional on every fixed key, its quantum marginal is still $\sigma_B$. Writing $c_B$ for the classical shares and $p(c_B)$ for their distribution, the joint state is
\[
\rho_{BC_B}(\psi)=\sum_{c_B}p(c_B)\,\ket{c_B}\bra{c_B}\otimes\sigma_B,
\]
which is independent of $\psi$ even when $c_B$ determines $(r,s)$.

\emph{Case 2.} If $B$ has a non-empty intersection with every pair, yet is still not authorized, then $B$ must contain exactly one qudit from each pair; hence, $|B|=n$.
This is the only scenario where residual leakage may occur. Up to permutation symmetry, we may write
\[
B=\{S_1,\dots,S_p,\;N_{p+1},\dots,N_n\},
\]
with $p$ signal qudits and $(n-p)$ noise qudits.

First, consider the reduced state when the input $\xi$ is encoded without randomization. From the qudit encrypted-cloning expression [Eq.~(\ref{eq:rho-s1s2np})], we have
\begin{equation}
\label{eq:rhoB-un}
    \rho_B^{\mathrm{(un)}}(\xi)=\frac{1}{d^n}\sum_{(a,b)\in\mathcal S}
\beta_{a,b}\,\langle \xi|X_d^aZ_d^b|\xi\rangle\,
\bigl[(X_d^aZ_d^b)^{\otimes p}\otimes (X_d^{-a}Z_d^b)^{\otimes (n-p)}\bigr],
\end{equation}
where the summation set is
\[
\mathcal S=\left\{(a,b)\in\mathbb Z_d^2\;\middle|\;
a+p b\equiv 0\pmod d,\;\; (n-p+1)a+b\equiv 0\pmod d
\right\},
\]
and $\beta_{a,b}:=\exp\!\left[-\frac{{\rm i}\pi}{d}\bigl(a^2+b^2+2ab(n-p)+\epsilon_d(a+b)\bigr)\right]$. The term $(a,b)=(0,0)$ gives $I_d^{\otimes n}/d^n$, which carries no information.

In the actual scheme, the encoded state is built from $|\psi_{r,s}\rangle=V_{r,s}|\psi\rangle$ with uniformly random $V_{r,s}$. Since $B$'s participants do not know $(r,s)$, their effective quantum state is the average over $r,s$:
\begin{equation}
\label{eq:bar_rho_b}
   \bar{\rho}_B=\frac{1}{d^2}\sum_{r,s\in\mathbb Z_d}\rho_B^{\mathrm{(un)}}(V_{r,s}|\psi\rangle).
\end{equation}

Substituting Eq.~(\ref{eq:rhoB-un}) into Eq.~(\ref{eq:bar_rho_b}) and interchanging sums yields
\begin{equation}
    \label{eq:rho-b-bar}
    \bar{\rho}_B=\frac{1}{d^n}\sum_{(a,b)\in\mathcal S}
\beta_{a,b}\,
\left[\frac{1}{d^2}\sum_{r,s}
\langle\psi|V_{r,s}^\dagger (X_d^aZ_d^b)V_{r,s}|\psi\rangle\right]
\bigl[(X_d^aZ_d^b)^{\otimes p}\otimes (X_d^{-a}Z_d^b)^{\otimes (n-p)}\bigr].
\end{equation}

For each $(a,b)\in\mathcal S$, we evaluate the bracketed expectation. If $(a,b)=(0,0)$, it equals $1$. If $(a,b)\neq(0,0)$, by Lemma~\ref{lem:weyl}, we obtain
\[
\frac{1}{d^2}\sum_{r,s}V_{r,s}^\dagger (X_d^aZ_d^b)V_{r,s}=0,
\]
so the expectation vanishes. Thus all nontrivial terms in Eq.(\ref{eq:rho-b-bar}) disappear, and we are left with
\[
\bar{\rho}_B=\frac{1}{d^n}I_d^{\otimes n}.
\]
Hence the averaged quantum state is maximally mixed, independent of $|\psi\rangle$.

In this second case $|B|=n$. By the privacy property of Shamir's scheme, $C_B$ is independent of $(r,s)$; conditioning on any value of $C_B$ therefore leaves $(r,s)$ uniformly distributed on $\mathbb Z_d^2$. Consequently the full joint state is $\tau_{C_B}\otimes I_d^{\otimes n}/d^n$, where $\tau_{C_B}$ is independent of $|\psi\rangle$.

Combining both cases, we conclude that every unauthorized subset obtains no information about the input state. The theorem is proved.
\end{proof}
\section{Conclusion}
\label{sec:conclusion}
In this work, we provided a complete classification of informative subsets of the storage register in the qudit encrypted-cloning protocol. We focused on the non-trivial case of unauthorized subsets of size $n$ that contain exactly one qudit from each signal-noise pair. By reducing the full encoded state to such subsets and exploiting the algebraic structure of generalized Pauli operators, we derived a closed-form expression for the reduced density matrix.

Our analysis showed that leakage is not generic but occurs only under precise arithmetic conditions. Specifically, for an aligned subset consisting of $p$ signal qudits and
$(n-p)$ noise qudits, we found that the subset is completely uninformative if and only if
$\gcd\!\left(d,p(n-p+1)-1\right)=1$; in that case the reduced state is maximally mixed. Otherwise, the subset retains a residual dependence on the input state through specific generalized Pauli operators, i.e., it is partially informative.

These results extended the parity-based classification known for qubits ($d=2$) to arbitrary finite dimensions and revealed a dimension-dependent confidentiality boundary in encrypted cloning. In addition, group-average orthogonality combined with classical secret sharing removes all input dependence from the joint quantum--classical state of every unauthorized subset. The quantum marginal is maximally mixed for aligned unauthorized subsets of size $n$; other unauthorized subsets are guaranteed to be input-independent but need not be maximally mixed.
This work provided a rigorous foundation for assessing information leakage in qudit-based redundant quantum storage and identified the algebraic origin of the leakage in the interference of phase factors and generalized Pauli commutation relations. Future directions include extending the classification to subsets that also contain the original input qudit and exploring modified encoding unitaries that may suppress leakage entirely.

\section*{Acknowledgments}
The authors are very grateful to anonymous reviewers for constructive comments that have greatly helped to improve the quality of this paper. This research was supported by the Nation Natural Science Foundation of China under Grant No.12301590 and Science Research Project of Hebei Education Department under Grant No.BJ2025061.
\bibliographystyle{unsrt}
\bibliography{main}

\begin{thebibliography}{10}

\bibitem{wootters1982single}
William~K Wootters and Wojciech~H Zurek.
\newblock A single quantum cannot be cloned.
\newblock {\em Nature}, 299(5886):802--803, 1982.

\bibitem{dieks1982communication}
Dennis Dieks.
\newblock Communication by epr devices.
\newblock {\em Physics Letters A}, 92(6):271--272, 1982.

\bibitem{bennett1992quantum}
Charles~H Bennett, Gilles Brassard, and Artur~K Ekert.
\newblock Quantum cryptography.
\newblock {\em Scientific American}, 267(4):50--57, 1992.

\bibitem{gisin2002quantum}
Nicolas Gisin, Gr{\'e}goire Ribordy, Wolfgang Tittel, and Hugo Zbinden.
\newblock Quantum cryptography.
\newblock {\em Reviews of modern physics}, 74(1):145, 2002.

\bibitem{pirandola2020advances}
Stefano Pirandola, Ulrik~L Andersen, Leonardo Banchi, Mario Berta, Darius
  Bunandar, Roger Colbeck, Dirk Englund, Tobias Gehring, Cosmo Lupo, Carlo
  Ottaviani, et~al.
\newblock Advances in quantum cryptography.
\newblock {\em Advances in optics and photonics}, 12(4):1012--1236, 2020.

\bibitem{renner2008security}
Renato Renner.
\newblock Security of quantum key distribution.
\newblock {\em International Journal of Quantum Information}, 6(01):1--127,
  2008.

\bibitem{li2025microsatellite}
Yang Li, Wen-Qi Cai, Ji-Gang Ren, Chao-Ze Wang, Meng Yang, Liang Zhang,
  Hui-Ying Wu, Liang Chang, Jin-Cai Wu, Biao Jin, et~al.
\newblock Microsatellite-based real-time quantum key distribution.
\newblock {\em Nature}, 640(8057):47--54, 2025.

\bibitem{Wiesemann2026consolidated}
Jerome Wiesemann, Jan Krause, Devashish Tupkary, Norbert L{\"{u}}tkenhaus,
  Davide Rusca, and Nino Walenta.
\newblock A consolidated and accessible security proof for finite-size
  decoy-state quantum key distribution.
\newblock {\em {Quantum}}, 10:2037, 2026.

\bibitem{wang2026demonstration}
Ke~Wang, Zhide Lu, Chuanyu Zhang, Gongyu Liu, Jiachen Chen, Yanzhe Wang, Yaozu
  Wu, Shibo Xu, Xuhao Zhu, Feitong Jin, et~al.
\newblock Demonstration of low-overhead quantum error correction codes.
\newblock {\em Nature Physics}, 22:308--314, 2026.

\bibitem{aydin2026quantum}
Arda Aydin, Victor~V Albert, and Alexander Barg.
\newblock Quantum error correction beyond su (2): spin, bosonic, and
  permutation-invariant codes from convex geometry.
\newblock {\em PRX Quantum}, 7(1):010341, 2026.

\bibitem{buvzek1996quantum}
Vladimir Bu{\v{z}}ek and Mark Hillery.
\newblock Quantum copying: Beyond the no-cloning theorem.
\newblock {\em Physical Review A}, 54(3):1844, 1996.

\bibitem{hillery1997quantum}
Mark Hillery and V~Bu{\v{z}}ek.
\newblock Quantum copying: Fundamental inequalities.
\newblock {\em Physical Review A}, 56(2):1212, 1997.

\bibitem{duan1998probabilistic}
Lu-Ming Duan and Guang-Can Guo.
\newblock Probabilistic cloning and identification of linearly independent
  quantum states.
\newblock {\em Physical review letters}, 80(22):4999, 1998.

\bibitem{bi2025virtual}
Zhi-Hao Bi, Jing-Tao Qiu, and Xiao-Dong Yu.
\newblock Virtual cloning of quantum states.
\newblock {\em Physical Review A}, 112(1):012420, 2025.

\bibitem{yamaguchi2026encrypted}
Koji Yamaguchi and Achim Kempf.
\newblock Encrypted qubits can be cloned.
\newblock {\em Physical Review Letters}, 136(1):010801, 2026.

\bibitem{huang2024performing}
Jose Luis~Lo Huang and Vincent~C Emeakaroha.
\newblock Performing distributed quantum calculations in a multi-cloud
  architecture secured by the quantum key distribution protocol.
\newblock {\em SN Computer Science}, 5(4):410, 2024.

\bibitem{hillery1999quantum}
Mark Hillery, Vladim{\'\i}r Bu{\v{z}}ek, and Andr{\'e} Berthiaume.
\newblock Quantum secret sharing.
\newblock {\em Physical Review A}, 59(3):1829, 1999.

\bibitem{zhang2025device}
Qi~Zhang, Jia-Wei Ying, Zhong-Jian Wang, Wei Zhong, Ming-Ming Du, Shu-Ting
  Shen, Xi-Yun Li, An-Lei Zhang, Shi-Pu Gu, Xing-Fu Wang, et~al.
\newblock Device-independent quantum secret sharing with advanced random key
  generation basis.
\newblock {\em Physical Review A}, 111(1):012603, 2025.

\bibitem{di2025security}
Alessio Di~Santo, Walter Tiberti, and Dajana Cassioli.
\newblock Security and fairness in multiparty quantum secret sharing protocol.
\newblock {\em IEEE Transactions on Quantum Engineering}, 6:1--18, 2025.

\bibitem{usenko2026continuous}
Vladyslav~C Usenko, Antonio Ac{\'\i}n, Romain All{\'e}aume, Ulrik~L Andersen,
  Eleni Diamanti, Tobias Gehring, Adnan~AE Hajomer, Florian Kanitschar,
  Christoph Pacher, Stefano Pirandola, et~al.
\newblock Continuous-variable quantum communication.
\newblock {\em Reviews of Modern Physics}, 98(1):015003, 2026.

\bibitem{cear2026clon}
F.~I. Cear{\v{a}}.
\newblock Cloning encrypted quantum states in arbitrary dimensions.
\newblock {\em Physical Review A}, 114(2):022450, 2026.

\bibitem{gianini2026encrypted}
Gabriele Gianini, Omar Hasan, Corrrado Mio, Stelvio Cimato, and Ernesto
  Damiani.
\newblock Encrypted clones can leak: Classification of informative subsets in
  quantum encrypted cloning.
\newblock {\em arXiv preprint arXiv:2604.10155}, 2026.

\bibitem{yamaguchi2026experimental}
Koji Yamaguchi, Leon Rullk{\"o}tter, Ibrahim Shehzad, Sean~J. Wagner, Christian
  Tutschku, and Achim Kempf.
\newblock Experimental demonstration that qubits can be cloned at will, if
  encrypted with a single-use decryption key.
\newblock {\em arXiv preprint arXiv:2602.10695}, 2026.

\bibitem{lim2026ame}
Zheng~Liang Lim and Hoi-Kwong Lo.
\newblock Encrypted cloning, absolute maximal entanglement and quantum secret
  sharing.
\newblock {\em arXiv preprint arXiv:2605.26866}, 2026.

\bibitem{gianini2026full}
Gabriele Gianini, Stelvio Cimato, Jianyi Lin, Omar Hasan, and Ernesto Damiani.
\newblock Full characterization of informative subsets in quantum encrypted
  cloning.
\newblock {\em arXiv preprint arXiv:2605.27421}, 2026.

\bibitem{gianini2026access}
Gabriele Gianini, Stelvio Cimato, Jianyi Lin, Omar Hasan, Corrado Mio, and
  Ernesto Damiani.
\newblock Beyond the canonical protocol: Quantum encrypted cloning from
  secret-sharing access structures.
\newblock {\em arXiv preprint arXiv:2606.06552}, 2026.

\bibitem{gianini2026diagnostic}
Gabriele Gianini, Omar Hasan, Stelvio Cimato, and Ernesto Damiani.
\newblock Encrypted redundancy as a diagnostic resource: Relational diagnosis
  in quantum encrypted cloning.
\newblock {\em arXiv preprint arXiv:2609.25043}, 2026.

\bibitem{roy2026resources}
Pritam Roy and Shashank Gupta.
\newblock Beyond maximal entanglement: Exact resources for multiparty encrypted
  quantum cloning.
\newblock {\em arXiv preprint arXiv:2609.28241}, 2026.

\bibitem{schwinger1960unitary}
Julian Schwinger.
\newblock Unitary operator bases.
\newblock {\em Proceedings of the National Academy of Sciences},
  46(4):570--579, 1960.

\bibitem{Shamir}
Adi Shamir.
\newblock How to share a secret.
\newblock {\em Commun. ACM}, 22(11), 1979.

\end{thebibliography}
\end{document}